\documentclass[10pt,final,twocolumn]{IEEEtran}

\IEEEoverridecommandlockouts                              
\usepackage{srcltx} 
\usepackage{amsmath,amsthm,amsxtra,amssymb}
\usepackage{times,latexsym,array,verbatim,url}
\usepackage{epsfig,color,graphicx,cite}

\newtheorem{thm}{Theorem}[section]

\newtheorem{prop}{Proposition}
\newtheorem{assm}{Assumption}

\newtheorem{prob}{Problem}

\theoremstyle{remark}
\newtheorem{rem}{Remark}

\newcommand{\Real}{\mathbb R}

\newcommand{\T}{\mathcal{T}}
\newcommand{\Th}{\hat{\mathcal{T}}}
\newcommand{\x}{\mathbf{x}}
\newcommand{\xh}{\hat{x}}

\newcommand{\G}{\mathcal{G}}
\newcommand{\U}{\mathcal{U}}
\renewcommand{\O}{\Omega}
\newcommand{\Oh}{\hat{\Omega}}

\newcommand{\sC}{\mathcal{C}}
\newcommand{\xtb}{\tilde{x}}
\renewcommand{\a}{\alpha}
\newcommand{\ah}{\hat{\alpha}}

\newcommand{\at}{\tilde{\alpha}}
\newcommand{\cb}{\mathbf{c}}
\newcommand{\bv}{\mathbf{b}}
\renewcommand{\sb}{\bar s}

\title{\LARGE \bf An Optimization and Control Theoretic Approach \\ to Noncooperative Game Design
}
 
\author{Tansu Alpcan, Lacra Pavel and Nem Stefanovic
\thanks{T. Alpcan is with the Deutsche Telekom Laboratories, Technical University Berlin, 10587 Berlin, Germany.
        {\tt\small alpcan@sec.t-labs.tu-berlin.de}}%
\thanks{L. Pavel and N. Stefanovic are with the department of Electrical and Computer Engineering, University of Toronto, Toronto, ON, Canada.
        {\tt\small pavel@control.utoronto.ca, nem@control.utoronto.ca}}%
}

\begin{document}

\maketitle
\thispagestyle{empty}
\pagestyle{empty}


\begin{abstract}
This paper investigates design of noncooperative games from an optimization and control theoretic perspective. 
Pricing mechanisms are used as a design tool to ensure that the Nash equilibrium of a fairly general class of noncooperative 
games satisfies certain global objectives such as welfare maximization or achieving a certain level of quality-of-service (QoS). 
The class of games considered provide a theoretical basis for decentralized resource allocation and control 
problems including network congestion control, wireless uplink power control, and optical power control.

The game design problem is analyzed under different knowledge assumptions (full versus limited information) and 
design objectives (QoS versus utility maximization) for separable and non-separable utility functions. 
The  ``price of anarchy'' is shown not to be an inherent feature of full-information games that incorporate pricing mechanisms. 
Moreover, a simple linear pricing is shown to be sufficient for design of Nash equilibrium according to a chosen global 
objective for a fairly general class of games.  Stability properties of the game and pricing dynamics are studied under the 
assumption of time-scale separation  and in two separate time-scales. 
Thus, sufficient conditions are derived, which allow the designer to place the Nash equilibrium solution or 
to guide the system trajectory  to a desired region or point. The obtained results  are illustrated with a number of examples.
\end{abstract}

\section{Introduction} \label{sec:intro}

Game theory has been used extensively as a quantitative framework for studying communication networks and distributed control systems among its other applications in engineering and economics. Game theoretic models provide not only a basis for analysis but also for design of network protocols and decentralized control schemes~\cite{basargame,srikantbook}. Despite widespread use of  (noncooperative)
game theory in engineering, there is surprisingly little work on \textit{how to design games} such that their outcome satisfies certain global objectives. 

While there is a general agreement on the usefulness of game theory, the issues of \textbf{price of anarchy} or \textit{efficiency loss} associated with noncooperative games even under the existence of pricing mechanisms,
have been  the subject of many investigations ~\cite{efficiencyloss1,efficiencyloss2,priceanarchy}. Consequently, 
different pricing schemes have been proposed in the literature aiming to improve Nash equilibrium (NE) efficiency with respect to a chosen criterion in specific settings~\cite{saraydar1,saraydar2,srikantbook,gamecomm1, ifac08-qzhulacra}. The research community has revisited the issue of game design again very recently~\cite{johari1,hajek1}. These studies are limited either to special problem formulations or adopt specific efficiency criteria such as the ``system problem'' of~\cite{kelly1}. 
A related early line of work focuses specifically on three agent (player) dynamic noncooperative games with multi-levels of hierarchy~\cite{basar-3agentgame1,basar-3agentgame2}, where it has been shown that the leader has an optimal incentive policy which is linear in the partial dynamic measurement and induces the desired behavior on the two followers. In addition, a separate but substantial literature exists under the umbrella of implementation theory, especially in the field of economics, which focuses on finding fundamental bounds for games where the outcome satisfies some given criteria~\cite{maskin1}. These works, however, are often not algorithmic and do not have an engineering perspective.


The \textbf{main contribution} of this paper is a treatment of game design from an optimization and control theoretic perspective. A fairly general class of games are considered, which have been applied to a number of settings including network congestion control, wireless uplink power control, and optical power control~\cite{alpcan-twc,ifac08-lacra,alpcancdc04,tac06-pavel,cdc07-yanlacra}. This paper -to the best of our knowledge- constitutes one of the few efforts aiming to investigate the general problem of game design in a constructive manner \textit{from an optimization and control theoretic perspective}

While it is straightforward to optimize NE according to some criterion under full information, the problem is much more complicated under information and communication constraints. The game or system designer (Figure~\ref{fig:nedesign1}) usually does not have full information about the system parameters such as user preferences or utility functions. Under this kind of information constraints, the designer either deploys additional dynamic feedback mechanisms or requires side information from the system depending on the specific design objectives. An example for the former case is a dynamic pricing scheme operating as an ``outer feedback loop''. If the objective is to achieve a social optimum (e.g. maximization of sum of user utilities) or satisfying some quality of service (QoS) conditions, then the designer often needs limited but accurate (honest) information from users or the system. 

It is important to note that if the users have the capability of manipulating such side information, then the design problem can be more involved or even ill-defined. For example, the goal of reaching a social optimum without knowing true user utilities but having only access to manipulated data may not be a realistic or even feasible one~\cite{maskin1}. Although mechanisms such as VCG have been proposed to circumvent these issues, the resulting schemes are often limited and demanding in terms of communication requirements~\cite{johari1}. 
As a starting point, the treatment in this paper is  restricted to a class of games where \textbf{players do not manipulate the game} by deceiving the system designer. Such players are sometimes called ``\textit{price-taking}'' as opposed to manipulative or ``price-anticipating''.
Thus, the utility functions are assumed here to accurately reflect player preferences.

\begin{figure}[htp]
\centering
\includegraphics[width=0.9 \columnwidth]{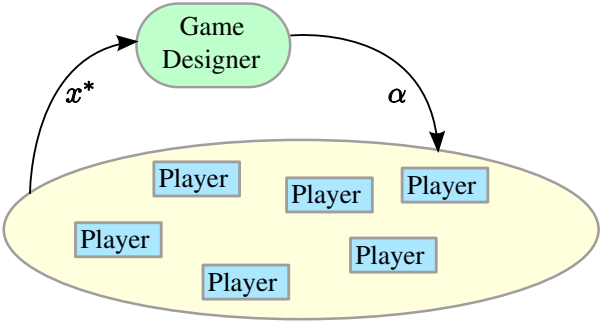}
\caption{The rules or pricing mechanisms within a game can be set by a ``designer'' to influence the outcome.
\label{fig:nedesign1}}
\end{figure} 

As an \textbf{additional contribution}, this paper analyzes dynamic systems arising from game formulations and the general case of non-separable player utilities. In such cases of coupled utility functions, the utility of each player is affected by decisions  of other players. The control theoretic approach adopted here based on game dynamics  provides a more realistic model for a number of applications compared to a static optimization one. Stability properties of game and pricing dynamics are investigated under the assumption of time-scale separation and in two separate time-scales using Lyapunov theory and singular perturbation approach. The sufficient conditions derived, which ensure system stability, are illustrated with examples. 



The rest of the paper is organized as follows. The next section introduces the underlying game model and problem formulation.
Section~\ref{sec:static} studies static design of games. In Section~\ref{sec:dynamic}, a dynamic control of games is discussed. Section~\ref{sec:incomplete} discusses game design under incomplete information followed by a numerical simulation in Section~\ref{sec:simulation}. The paper concludes with remarks of Section~\ref{sec:conclusion}.

%
%
%

\section{Game Model and Problem Formulation} \label{sec:model}

Consider a generic class of  $N$ player static noncooperative games, denoted by $\mathbf{\G0}$, on the action (strategy) space $\O \subset \Real^N$. The actions of players, who are autonomous decision makers, are denoted by the vector $x \in \O$, with $x_i$ being the   $i^{th}$  player's action. Furthermore, the $i^{th}$ player is associated with a smooth (continuously differentiable) cost function, $J_i:\Real \times \O \rightarrow \Real$,  $J_i(\a_i,x),\; i=1,2,\ldots, N$, parametrized by a scalar ``pricing'' or game parameter $\a_i \in \Real$. In some formulations, there may be (coupled) restrictions on the domain of these parameters such that $\a \in \tilde \O \subset \Real^N$. 
Assuming a set of sufficient conditions for the existence of at least one Nash equilibrium (NE) are satisfied, 
define a game mapping, $\T$ (an inverse game mapping $\Th$) that maps game parameters $\a$ to NE points (NE points to parameters):
\begin{equation} \label{e:map1}
  \T\,:\; \Real^N \rightarrow \O \;\;\;\text{and}  \;\;\;  \Th\,:\;  \O \rightarrow \Real^N,
\end{equation}
such that 
\begin{equation} \label{e:map2}
  x^* = \T(\a^*) \;\;\;\text{and}  \;\;\;  \a^*=\Th (x^*)
\end{equation}
for any NE point $x^*$ and corresponding parameter vector $\a^*$. Notice that the mappings $\T$ and $\Th$ are highly nonlinear, often not explicitly expressible, and may not be one-to-one or invertible except for special cases, i.e. games with special properties.

We now define \textbf{a class of games}, $\G1$, by assuming a cost structure of the form
\begin{equation} \label{e:cost1}
J_i(\a_i,x)= \a_i p_i(x) - U_i(x),
\end{equation}
where $\a_i \geq 0\; \forall i$, the continuously differentiable functions $p_i$ and $U_i$ are convex and strictly concave with respect to $x_i$ for any given $x_{-i}$, respectively. Let in addition the
strategy space $\O$ to be convex, compact, and nonempty. Thus, a game belongs to class $\G1$, if
the following assumptions are satisfied:
\begin{assm} \label{assm1}
The strategy space $\O$ is convex, compact, and has a nonempty interior, $\O \neq \varnothing$.
\end{assm}
\begin{assm} \label{assm2}
The cost function of the $i^{th}$ player $J_i$ in (\ref{e:cost1}) is twice continuously 
differentiable in all its arguments and strictly convex in $x_i$, i.e.,
$\partial^2 J_i(x)/ \partial x_i^2 > 0 $.
\end{assm}
Then, under the Assumptions~\ref{assm1} and \ref{assm2} the game $\G1$ admits a Nash equilibrium
from Theorem 4.4 in~\cite[p.176]{basargame}.


Next, we consider \textbf{a specific class of games}, $\G2$, as a special case of $\G1$ with additional conditions on the cost structure, such that they admit a unique NE solution for each given $\a$. Toward this end,
define the pseudo-gradient operator
\begin{equation}\label{e:psgrad}
 \overline \nabla J:= \left [\partial J_1(x) / \partial x_1 \cdots
  \partial J_N(x) / \partial x_N \right ]^T  := q(x).
\end{equation}
Subsequently, let the $N \times N$ matrix $Q(x)$ be the Jacobian of $q(x)$
with respect to $x$:
\begin{equation}\label{e:g1}
 Q(x):=
   \begin{pmatrix}
   b_1 & a_{12} & \cdots & a_{1N} \\
   \vdots &   & \ddots & \vdots \\
   a_{N1} & a_{N2}   & \cdots     & b_N \
 \end{pmatrix} \;,
\end{equation}
where $b_i$ and $a_{ij}$ are defined as
$b_i:=\frac{\partial^2 J_i(x)}{\partial x_i^2}$ and
$a_{i,j} :=  \frac{\partial^2 J_i(x)}{\partial x_i \partial x_j}$,
respectively. The following assumptions have to be satisfied in addition
to Assumptions~\ref{assm1} and \ref{assm2}, in order for a game to
be in class $\G2$:
\begin{assm} \label{assm3}
The symmetric matrix $Q(x)+Q(x)^T$, where $Q(x)$ is defined in~(\ref{e:g1}),
is positive definite, i.e. $Q(x)+Q(x)^T >0$ for all $x \in \O$.
\end{assm}

\begin{assm} \label{assm4}
The strategy space $\O$ of the class of games $\G2$ can be described as 
\begin{equation} \label{e:xdef}
 \O:=\{x \in \Real^N\,:\, h_j(x) \leq 0 ,\; j=1,\,2,\,\ldots r\},
\end{equation} 
where $h_j: \Real^N \rightarrow \Real, j=1,\,2,\,\ldots r$, $h_j(x)$ is convex in its arguments
for all $j$, and the set $\O$ is bounded and has a nonempty interior. In addition, the derivative of 
at least one of the constraints with respect to $x_i$,
$\{d h_j(x) / d x_i ,\; j=1,\,2,\,\ldots r\}$, is nonzero for $i=1,\,2,\,\ldots N$, $\forall x \in \O$.
\end{assm}

If Assumptions \ref{assm1}-\ref{assm4} are satisfied, then the class of games $\G2$ admits a unique
Nash equilibrium. Furthermore, this unique NE can be an inner or a boundary solution. We refer
to Appendix as well as \cite{rosen,basargame,tansuphd} for the details and
an extensive analysis. Notice that a large set of network games belong to this class with 
notable examples of network congestion games~\cite{srikantbook,alpcan-elektrik}, power control 
games in wireless networks~\cite{alpcan-twc} and optical networks~\cite{tac06-pavel}. 

Assume that the utility function $U_i$ accurately reflects the preferences of users or players, who
are \textit{autonomous (independent) decision makers}. Then, the pricing function $p_i$ and parameters $\a$ 
are in this setting the \textit{only} tools for the system designer to influence (optimize or control) 
the game outcome in order to achieve a desired objective. This task described in
Problem~\ref{problem1} is similar -in spirit- to the goal of implementation theory or mechanism design in the economics literature~\cite{maskin1} with the important difference of not allowing users to knowingly manipulate the system, i.e.
assuming price-taking users only. This game design problem  is illustrated in Figure~\ref{fig:ne1}.
\begin{figure}[htp]
\centering
\includegraphics[width=0.7\columnwidth]{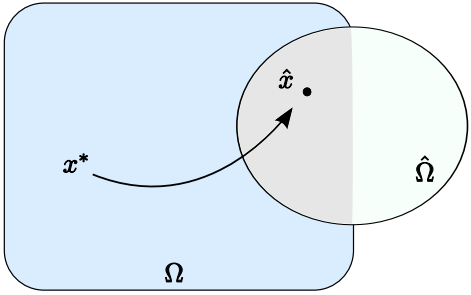}
\caption{Game design involves controlling the game dynamics such that the NE, $x^* \in \O$ is moved to a feasible desired region $\Oh \cap \O$ or a specific optimal point  $\xh$.
\label{fig:ne1}}
\end{figure} 

\begin{prob} \label{problem1}
How to choose the pricing function $p$ and parameters $\a$ such that the NE of games of class $\G1$ satisfies some desirable properties? 
\end{prob}
Two specific but common examples of such properties are
\begin{enumerate}
 \item The NE coincides with the solution of a global optimization problem, e.g. welfare maximization:
$$ x^*=\arg \max \sum_i U_i(x) \;\; \text{such that } x \in \O.\footnote{All summations in the paper are from $1,\ldots,N$ unless explicitly defined.}$$
\item The NE satisfies some system or user-dependent constraints such as capacity constraints, non-negativity, or performance bounds. For example, the favorable set $\Oh$ can be defined as 
$$\Oh:=\{x \in \Real^N : x_i\geq 0\; \forall i,\; \sum_i x_i \leq C,\; s_i(x)\geq \bar s_i \; \forall i \},$$
where $C$ is a capacity constraint, $s(\cdot)$ is a quality of service (QoS) measure such as signal-to-interference  ratio (SIR), and
$\bar s$ denotes the minimum acceptable QoS level.
\end{enumerate}

An important aspect of Problem~\ref{problem1} is the amount of knowledge available to the system designer 
in optimization of NE. If there is complete knowledge of player preferences and global system objective, then the approach to be adopted can be quite different from the one when the designer has very limited information. 
In cases when the game dynamics are very fast, it is appropriate to focus on static optimization of the NE point. Then, the actions of the system designer are assumed to be on a slower time-scale than the actual game dynamics between the players resulting in a hierarchical structure. This scenario is discussed in the next section.


\section{Static Design of Games} \label{sec:static}

When the game designer has  a way of estimating game parameters, the game design can be
posed as a static optimization problem, especially when the game dynamics are very fast and the actions of the system designer are on a slower time-scale than the actual game dynamics between the players. 
In contrast, Section~\ref{sec:dynamic} presents a dynamic control approach appropriate when the game dynamics are slow or there are external disturbances. Then, the game is treated as a control system that needs to be stabilized around a desired point.

When is it feasible to design a game such that the NE point can be located by the system designer to a point or region with desirable properties? In the point case, let target point be $\xh$. Then, the problem is to find an $\ah$ such that $\ah=\Th(\xh)$, for any desirable feasible $\xh$. The following surprisingly simple result addresses this problem for a fairly
general class of games.
\begin{thm} \label{thm:pricing}
For games of class $\G2$ with the cost structure given in~(\ref{e:cost1}) and under complete information assumption, affine pricing of the form, $\a\, p(\cdot)$, is sufficient to locate the unique NE point of the game to any desirable feasible point, $\xh \in \O$, as long as
$$
\dfrac{\partial p_i(\xh)}{\partial x_i} \neq 0, \;\; \forall i.
$$
\end{thm}
\begin{proof}
The proof immediately follows from the first order necessary optimality conditions of player cost optimization problems due to
the convexity of the cost structure and uniqueness of NE. 
$$
\a_i \dfrac{\partial p_i(\xh)}{\partial x_i} -  \dfrac{\partial U_i(\xh)}{\partial x_i} =0
\Rightarrow \ah_i=\left[ \dfrac{\partial p_i(\xh)}{\partial x_i}\right]^{-1}  \dfrac{\partial U_i(\xh)}{\partial x_i}\;\; \forall i
$$
for any feasible $\xh$.
\end{proof}

\begin{rem}
Theorem~\ref{thm:pricing} can easily be extended to the case where users actions are on a multi-dimensional subspace if the users utility function is separable.
\end{rem}

Notice that even a simple linear pricing function $p(x_i)=x_i$ satisfies the conditions of the theorem and is sufficient for NE optimization. In this case any  $\xh \in \O$ is feasible. However, a symmetric pricing scheme, where $\a_i=\a_j\; \forall i,j$, is not sufficient in general. As other examples,  for  $p(x_i)=e^{x_i}$ any  $\xh$ is feasible, while for  $p(x_i)=x^2_i$ any  $\xh \neq 0$ is feasible.

If a game admits multiple NE, e.g. games of class $\G1$, then reaching a single desirable point does not make much sense. Furthermore, the problem of locating all NE points to a desirable region can be rather complex. Such cases can be handled either by exploiting any  special structure of the game due to its specific problem domain or using numerical methods.

It is interesting to note that due to the form of the player cost structure in~(\ref{e:cost1}), i.e. additive linear pricing, the optimization problem of individual players in the game can be interpreted as a variant of \textit{Legendre transform}~\cite{legendre1}. In other words, given the other players' actions, the utility function best response of a player is transformed  to a best response function which takes the pricing parameter as argument.

Although tragedy of commons or price of anarchy are unavoidable in ``pure'' games without any external factors, they can be circumvented when additional mechanisms such as ``pricing'' are included in the game formulation. In parallel to some earlier results~\cite{rajiv1}, Theorem~\ref{thm:pricing} clearly establishes that \textit{``loss of efficiency'' or ``price of anarchy''} are not an inherent feature of a fairly general class of games with built-in pricing systems. 
If there is sufficient information, then any game of class $\G2$ can be designed through simple (e.g. linear) pricing mechanisms in such a way that any desirable criteria such as welfare maximization or QoS requirements are met at the unique NE solution. 





\subsection*{Example 1} \label{sec:expwireless}

In order to  illustrate the underlying issues in game design problems with accurate but limited information, consider a noncooperative
game formulation of the single-cell spread-spectrum uplink power control problem in wireless networks~\cite{winet}. 
Specifically, the game is played in a single cell with $N$ mobiles competing for quality of service in terms of 
signal-to-interference ratio (SIR) and the base station acts as the game designer.
This can be formulated as a class  $\G2$ game with the following general cost structure
\begin{equation} \label{e:expgameform1}
  J_i(\a_i,x)= \a_i p_i (x)  - U_i(x),\;\;  \forall i,
\end{equation}
where the objective is to locate the NE to a region that satisfies some feasibility and QoS constraints.
Further assume that the utility function $U_i$ of user $i$ in this game is
\begin{equation} \label{e:exSIR}
\begin{array}{l}
 U_i(x)=\beta_i \log(1+s_i(x)), \\ \\
\text{where } s_i(x):=\dfrac{h_i x_i}{\sum_{j \neq i} h_j x_j + \sigma^2}.
\end{array} 
\end{equation}
Here, $s$ represents the SIR with $h_i>0\; \forall i$ denoting channel gain parameters and $\sigma^2$ a noise term.
Then, each user $i$ decides on its own power level $x_i$ and is associated
with the cost function $J_i$  in~(\ref{e:expgameform1}).

The desired region for the NE of this game could be shaped by feasibility constraints such as positivity of user actions (here uplink transmission power levels) and an upper-bound on the sum of them, and/or some chosen minimum SIR levels (assuming these are chosen such that the region is not empty). A concrete example region $\Oh$ can be defined as
\begin{equation} \label{e:qosregion1}
 \Oh:=\{x \in \Real^N : x_i\geq 0,\; s_i(x)\geq \bar s_i \; \forall i \},
\end{equation}
where 
$\bar s_i$ are user-specific minimum SIR levels. 
A detailed analysis of an example case is provided next. For a separate but similar example of this formulation we refer to \cite{cdc07-yanlacra}.

As a special case, the pricing function $p_i(x)$ is chosen next to be linear in $x_i$ such
that  $$  J_i(\a_i,x)= \a_i x_i  - \beta_i \log(1+s_i(x)) ,$$
where $s_i(x)$ is defined in (\ref{e:exSIR}).
Under appropriate assumptions, the game is one of class $\G2$ and admits a unique inner NE solution, $x^*$.

For notational convenience, we define the matrix
\begin{equation} \label{e:matrix1}
A:= \begin{pmatrix}
  1 & \dfrac{h_2}{L h_1} & \dfrac{h_3}{L h_1} & \cdots & \dfrac{h_N}{L h_1} \\
  \dfrac{h_1}{L h_2} & 1 & \dfrac{h_3}{L h_2} & \cdots & \dfrac{h_N}{L h_2} \\
  \dfrac{h_1}{L h_3} & \dfrac{h_2}{L h_3}  & 1 & \cdots & \dfrac{h_N}{L h_3}  \\
   \vdots  &\vdots &   &  \ddots & \vdots \\
  \dfrac{h_1}{L h_N} &  \dfrac{h_2}{L h_N}  & \cdots & \dfrac{h_{N-1}}{L h_N}& 1
\end{pmatrix} 
\end{equation}
Then, the NE is the solution of
$$  A x^* = \cb, $$
where 
$$ \cb:=\left[ \frac{\beta_1}{\a_1} - \frac{\sigma^2}{L h_1}, \ldots,\frac{\beta_N}{\a_N} - \frac{\sigma^2}{L h_N}  \right].$$

The desired QoS region $\Oh$ in (\ref{e:qosregion1}) can be alternatively described  in terms of received power  levels at the
base stations and in matrix form~\cite{alpcan-winet2}:
$$ \Oh=\{x \in \Real^N : x_i\geq 0 \; \forall i, \; S x \geq \bv \}, $$
where the matrix $S$ is defined as
\begin{equation} \label{e:amatrix}
   S:=
   \begin{pmatrix}
   h_{1} & -h_{2}\dfrac{\sb_1}{L}& \cdots & \dfrac{-h_{N }\sb_1}{L} \\
   \dfrac{-h_{1 }\sb_2}{L} & h_{2 } & \cdots & \dfrac{-h_{N }\sb_2}{L} \\
   \vdots &   & \ddots & \vdots \\
    \dfrac{-h_{1 }\sb_{N}}{L} & \dfrac{-h_{2 }\sb_{N}}{L}   & \cdots   & h_{N} \\
 \end{pmatrix},
 \end{equation}
and
$$\bv:=\left[ \frac{\sb_1 \sigma^2}{L},  \ldots,  \frac{\sb_{N} \sigma^2}{ L} \right]^T .$$


If the designer, here the base station, has full information, then given a feasible target SIR level $\sb$
it is possible to solve for a pricing vector $\a$ such that the 
NE is on the boundary of $\Oh$, i.e. $S x = \bv$. This is due to both matrices $A$ and $S$ being nonsingular
as $h_i>0\; \forall i$. Hence, the appropriate pricing vector $\a$ can be immediately obtained from the
boundary solution
$$ \cb= A S^{-1}(\bv) , $$
and the definition of $\cb$.

\section{Dynamic Control of Games} \label{sec:dynamic}

Many games are solved in a distributed manner. Hence, the convergence of the system trajectory to the equilibrium point may not be very fast and the time-scale separation between system designers actions and actual game dynamics may fail. Then, the game design can be modeled as a \textbf{feedback control system} which utilizes pricing as the control input and the desired target as the reference (see Figure~\ref{fig:control1}). This formulation also brings a certain robustness with respect to initial conditions or game (system) parameters. The latter case is especially relevant for systems that are non-stationary over longer time periods and can also be formulated as a tracking problem. Specific examples are congestion and power control game formulations,~\cite{alpcan-twc,alpcan-elektrik,ifac08-qzhulacra}.
\begin{figure}[htp]
\centering
\includegraphics[width=0.9\columnwidth]{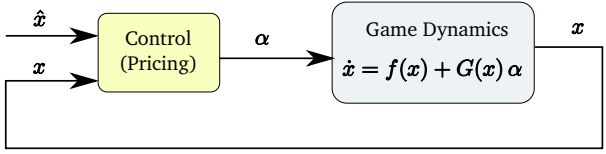}
\caption{Feedback control of the game (NE, $x^*$) using pricing $\a$ as the control parameter and $\xh$ as the desired reference signal.
\label{fig:control1}}
\end{figure} 

The counterpart of the feasibility question in the case of static game optimization relates  in the dynamic control setting to the \textbf{controllability} of the system shown in Figure~\ref{fig:control1}, or reachability of a state  $\xh$. In order to provide a concrete example to the problem of controllability, consider a game of class $\G2$ where the players adopt a gradient algorithm to optimize their own cost. Then, the game dynamics are:
\begin{equation} \label{e:system1}
   \dot{x}_i = - \dfrac{\partial J_i(x)}{\partial x_i}= \dfrac{\partial U_i(x)}{\partial x_i} - 
 \dfrac{\partial p_i(x)}{\partial x_i} \,\a_i \;\; \forall i,
\end{equation}
where $\a$ acts as the feedback control on the outcome of the game. Here, the objective is to investigate the conditions under which the game system is controllable. We write ~(\ref{e:system1}) in vector form as
\begin{equation} \label{e:system2}
   \dot{x} =  f(x) + \sum_{i=1}^{N} g_i(x) \, \a_i =  f(x) +  G(x) \, \a
 \end{equation}
where $\a = [\a_1 \dots \a_N]^T$, the matrix 
$$G(x) = \left [ g_1(x), \dots,  g_N(x) \right ] , $$ 
and the vector
$$
f(\x) = \left [
\begin{array}{lllll}
 \dfrac{\partial U_1(x)}{\partial x_1} & \dots & \dfrac{\partial U_i(x)}{\partial x_i} & \dots & \dfrac{\partial U_N(x)}{\partial x_N} \end{array} \right ]^T.
$$ Alternatively, the matrix $G$ is given by
$$
G(x) = \left [ \begin{array}{ccc} - \dfrac{\partial p_1(x)}{\partial x_1} & \dots & 0 \\ \dots & - \dfrac{\partial p_i(x)}{\partial x_i}  &\dots \\ 0 & \dots & -\dfrac{\partial p_N(x)}{\partial x_N} \end{array} \right ] .
$$

Based on the standard theorem on controllability using Lie brackets~\cite[Chapter 1]{isidoribook}, we obtain the following result. 

\begin{thm} \label{thm:controllab}
For games of class $\G2$ with cost structure~(\ref{e:cost1}) and  game dynamics ~(\ref{e:system1}), or ~(\ref{e:system2}),  a sufficient condition for local reachability around a point $x_0$ is that the distribution $\sC$ satisfies the rank condition at $x_0$, $dim \sC (x_0)=N$ where
$$
\sC = \left [ g_1, \dots,g_N, [g_i,g_j], \dots, [f,g_i], \dots, ad^{k}_f g_i,\dots  \right ]
$$
where $[f,g_i](x) = \dfrac{\partial g_i(x)}{\partial x} f(x) -   \dfrac{\partial f(x)}{\partial x} g_i(x)$ is the Lie bracket of  $f$ and $g_i$, and $ad^{k}_f g_i$ denote higher order Lie brackets defined recursively by $ad^{k}_f g_i(x) = [f, ad^{k-1}_f g_i](x)$.
\end{thm}

\begin{rem}
Notice that if the diagonal matrix $G(x_0)$ has rank $N$, any $\xh$ locally around $x_0$ is reachable in finite time under piecewise constant input functions, which is equivalent to the feasibility condition in Theorem~\ref{thm:pricing}. In addition, for the simple linear pricing function $p(x_i)=x_i$  any $\xh$ is   immediately reachable since  $G(\xh)$ is constant and invertible.
\end{rem}

Alternatively, the problem can be posed as one of  asymptotic set point regulation, i.e., to find a feedback control of the form $\a=\a(x, \xh)$ such that  the system trajectory $x$ (and eventually the NE) converges to the desired reference point $\xh$~\cite[Chapter 8]{isidoribook}. Considering the game system ~(\ref{e:system2}) and  the change of variables 
$\xtb = x - \xh$, in the new coordinates the game system becomes
\begin{equation} \label{e:system3}
   \dot{\xtb} =  f(\xtb + \xh) + G(\xtb + \xh) \, \a
 \end{equation}
and the design problem is to find control $\a$ to stabilize the equilibrium $\xtb = 0$. The first necessary condition, which translates into feasibility condition for $\xh$ is that ~(\ref{e:system3}) has an equilibrium at the origin, i.e. there exists a steady-state control  $\a_s = c(\xh)$ that solves
\begin{equation} \label{e:FBI}
 0 = f(\xh) + G(\xh) \, c(\xh)
\end{equation}
The component  $\a_s$ is the first component in $\a$ needed to maintain equilibrium at the origin, while a second component is needed to asymptotically stabilize this equilibrium in the first approximation.   Thus,  based on the necessary and sufficient conditions for asymptotic regulation problem (see Theorem 8.3.2 in~\cite{isidoribook}), specialized for constant reference, a feedback control law  for the pricing parameter that solves this problem in the full information case is $\a = \a(x, \xh)$  with
\begin{equation} \label{e:ctrl}
\a(x,\xh) = c(\xh) + K \, (x - \xh)
\end{equation} 
where $ c(\xh)$ solves ~(\ref{e:FBI}) and  $K$ selected such that $(A+B \, K)$ has eigenvalues with negative real part. The matrices $A$ and $B$ denote the Jacobians of $f$ and $g$ at the origin, respectively. Using ~(\ref{e:system2}) it can be seen that the feasibility condition ~(\ref{e:FBI}) is equivalent to feasibility in the static case (Theorem ~\ref{thm:pricing}). Also, for $G(\xh)$ full rank this feasibility condition of the regulation problem corresponds to  $\xh$ being reachable (see remark after Theorem ~\ref{thm:controllab}). Note also, that  for constant reference $\xh$, integral control could be included in the design formulation by augmenting the system ~(\ref{e:system2}) with a stack of N integrators
\begin{equation} \label{e:augm}
 \dot{\sigma} =  \xtb
\end{equation}
The  task is to design a stabilizing feedback controller $\a = \a(\xtb, \sigma)$  that stabilizes the augmented state model ~(\ref{e:system3},\ref{e:augm}) at the equilibrium point $(0, \sigma_s)$ , where $\sigma_s$ produces the desired steady-state control~\cite{isidoribook}.

\subsection*{Example 2}

Consider the power control problem in optical networks \cite{tac06-pavel} with linear pricing and an optical signal-to-noise ratio (OSNR)-like utility function. Similar to the wireless case, the players here compete for quality of service in terms of OSNR. and are 
coupled due to interference. The user $i$ chooses own power level $x_i$ based on the following cost function
\begin{equation} \label{e:OSNRp}
J_i(x) = \alpha_i x_i - \beta_i \log \left( 1+ a_i \frac{\gamma_i(x)}{1 -  \Gamma_{ii} \gamma_i(x) } \right)
\end{equation}
a special case of (\ref{e:cost1}), where  $\beta_i$, $a_i$ are design parameters and $\gamma_i(x)$ denotes the OSNR
\begin{equation}\label{e:OSNRdef}
\gamma_i(x) = \frac{x_i}{n_{0} + \sum_{j} \Gamma_{ij} x_j } 
\end{equation}
for a given system matrix $\Gamma$ and input noise $n_0$. Let $\tilde{\Gamma}$ be the matrix with elements 
\begin{equation}
\tilde{\Gamma}_{ij} = \left\{ \begin{array}{c} a_i, \hspace{0.6em} \quad i=j \\ \Gamma_{ij}, \quad i \neq j \end{array} \right. . \label{e:gamma}
\end{equation}
Then the corresponding dynamic system, similar to~(\ref{e:system2}), is
$$
\dot{x} = f(x) - \a
$$
where $f(x)$ is the vector with elements
$$ \label{e:fx}
\frac{\partial U_i(x)}{\partial x_i} = \frac{a_i \beta_i}{ n_{0} + \sum_{j} \tilde{\Gamma}_{ij} x_j  }
$$
or, in shifted coordinates, $\xtb = x - \xh$,
$$
\dot{\xtb} = f(\xtb + \xh) - \a
$$
Linearizing around the origin, $\xtb = 0$, yields $B=-I$ and
$$
A = \left. \frac{\partial f(x)}{\partial x} \right|_{x=\xh} = D \tilde{\Gamma} 
$$
with  $\tilde{\Gamma}$ as in~(\ref{e:gamma}) and
$$
D = diag\left( \frac{-a_i \beta_i}{(n_{0}+ \sum_{j} \tilde{\Gamma}_{ij} \xh_j)^2} \right).
$$
An appropriate gain $K$ can be selected to ensure stability at the origin for $A+BK$ if the feasibility conditions are satisfied,~\cite{isidoribook}, i.e., $(A,B)$ is stabilizable and   
$$
\left[ \begin{array}{cc} A & B \\ I & 0 \end{array} \right] = \left[ \begin{array}{cc} D\tilde{\Gamma} & -I \\ I & 0 \end{array} \right]
$$
is nonsingular, which is clearly the case. Thus, tracking is achievable, and pole placement methods can be applied to derive an exact gain parameter $K$. In addition, the system can be made more robust by applying integral control as an alternative. The same two feasibility conditions that apply to feedforward control, which our system $(A,B)$ already satisfies, also apply to integral control. The augmented linearized system $(A,B)$ with integral control has the form
$$
\left[ \begin{array}{c} \dot{\xtb} \\ \dot{\sigma} \end{array} \right] = \left[ \begin{array}{cc} D \tilde{\Gamma} & 0 \\ I & 0 \end{array} \right] \left[ \begin{array}{cc} \xtb \\ \sigma \end{array} \right] + \left[ \begin{array}{cc} -I \\ 0 \end{array} \right] \alpha 
$$
with control law $\alpha =  K \xtb + K_I \sigma$. The gains $K$ and $K_I$ can be determined using pole placement applied to the augmented system. 

\section{Game Design under Information Constraints} \label{sec:incomplete}


Unlike the case discussed in the previous section, the system designer usually does not have full information about the system parameters such as user preferences or utility functions. Under such \textbf{information constraints}, the designer either deploys additional dynamic feedback mechanisms or requires side information from the system, depending on the specific design objectives. An example for the former case is a dynamic pricing scheme operating as an ``outer feedback loop''. If the objective is to achieve a social optimum (e.g. maximization of sum of user utilities) or satisfying some QoS constraints, then the designer often needs \textit{accurate} side information from users or the system. Assuming that users are non-manipulative (honest or price-taking) and given accurate side information, the task of the designer can be formulated as an optimization problem even if it is solved indirectly or in a distributed manner. Next, an example formulation is provided that illustrates the underlying concepts. T\textit{he objective is let the NE coincide with a social optimum (maximizing sum of user utilities) in the general case of non-separable user utilities of the form $U_i(x)$.}

\subsection{Pricing Dynamics under Time-Scale Separation}

Define a strictly concave and smooth social welfare function $ \U(x)$ which is a sum of concave and \textbf{non-separable utility} functions $ \U(x):=\sum_i U_i(x)$ and admits a global maximum $ \xh=\arg \max_{x} \sum_i U_i(x)$. As before  the cost function is $ J_i(\a_i,x)= \a_i p_i (x)  - U_i(x)$. Here, unlike the separable one in~\cite{gamenetsne}, side information (e.g. $U_i(x^*)$) is required in order to bring the NE to the social maximum point. The social maximum is defined easily via the first order optimality conditions
$$ \dfrac{\partial \U}{\partial x}(\xh)=0,$$
where 
$$ \dfrac{\partial \U}{\partial x}(x) = \left [ \begin{array}{ccc}  \sum_j \dfrac{\partial U_j}{\partial x_1}(x)   & \dots &  \sum_j \dfrac{\partial U_j}{\partial x_N} (x)  \end{array} \right ]. $$
The social maximum is shown to coincide with the unique equilibrium (and NE) of the following pricing mechanism 
\begin{equation} \label{e:pricing2}
  \dot \a_i = \sum_j  \sum_k \dfrac{\partial U_j}{\partial x_k^*} \dfrac{\partial x_k^*}{\partial \a_i}  \;\; \forall i
\end{equation}
If these pricing dynamics are on a slower time scale than the game dynamics, then the system designer can obtain sufficiently accurate estimates of $\partial U_i(x^*) / \partial x_i$ and $\partial x_i^* / \partial \a_i$. 

As one possibility, if the users adopt a gradient algorithm, then the designer can use past values of $x^*$ and $\a$ along with the exact form of the pricing functions $p$  to estimate $\partial U_i(x^*) / \partial x_i$ directly without requiring any side information (except from some fixed system parameters). In addition, side information (e.g. $U_i(x^*)$) is required to estimate $\partial U_i(x^*) / \partial x_j$ and $\partial x_i^* / \partial \a_j$ for all $i, \, j$. 

Assume an ideal case where the parameter estimation is perfectly accurate. Then, the pricing mechanism above ensures that the NE point of the underlying game globally asymptotically converges to the maximum of the social welfare function.\footnote{For simplicity,  the social maximum point is implicitly assumed to be on the solution space of the game.} The next theorem follows  from Lyapunov theory and LaSalle's theorem in a straightforward manner. The Lyapunov function is chosen to be the negative of social welfare function, $V = -\U$.

Assume that the pricing mechanism is on a slower time scale than the actual game dynamics leading to a time-scale separation, and hence to a hierarchically structured problem. Assuming this \textbf{time-scale separation}, for simplicity, initially only the pricing dynamics (slower dynamics) is  considered.  

\begin{thm} \label{thm:sepconv}
Consider a class $\G2$ game with cost structure~(\ref{e:cost1}) and user utilities $U_i\;\forall i$, where
the objective function  $\U(x):=\sum_i U_i(x)$ admits a unique inner global maximum $ \xh=\arg \max_{x} \U(x)$.
Then, the pricing mechanism~(\ref{e:pricing2}) ensures that the NE point of the underlying game, $x^*$, globally asymptotically converges to the maximum of the social welfare function, $\xh$, if the Jacobian matrix of  the mapping $\T$ with respect to pricing vector $\a$, defined as
$$H(\a) := \dfrac{\partial x^*}{\partial \a}(\a) = \left [  \dfrac{\partial x^*_i}{\partial \a_j}(\a) \right ], \;\; i,j=1,\dots,N, $$ 
is non-singular.
\end{thm}
\begin{proof}
The pricing dynamics (slow) are analyzed, assuming that the user dynamics is fast and converges quickly   to  $x^* = \T(\a)$ for any given $\a$. The pricing scheme~(\ref{e:pricing2}) admits a \textbf{unique equilibrium} \textit{if and only if} 
$\partial \U(x) / \partial x=0$.\footnote{We drop in the rest of the proof the $(\cdot)^*$ notation, which characterizes the NE, for convenience.} 
In order to prove this statement,  first change the summation order in~(\ref{e:pricing2}) and write it in vector form as
\begin{equation} \label{e:pricing1v}
  \dot \a = H^T(\a) \; \left [\dfrac{\partial \U}{\partial x}(x) \right ]^T
\end{equation}
where $x = \T(\a)$. The \textbf{sufficiency} immediately follows from non-singularity of $H$. If $\partial \U / \partial x=0$, then $\dot \alpha =0$ and at the same time the system is at the social maximum $\xh$. Hence, at the unique equilibrium of the pricing scheme, $\ah$, the social maximum is achieved.
Next, the \textbf{necessity} is shown under the same condition.  At the unique equilibrium point $\ah$ of the pricing scheme~(\ref{e:pricing1v}),  it is necessary that  $\partial \U / \partial x=0$ due to non-singularity of of $H$. Consequently, at the unique equilibrium point $\ah$ of the pricing scheme the objective function $\U(x)$ necessarily reaches its maximum $\xh$ characterizing the social maximum, hence $\xh = \T(\ah)$ or $\ah = \Th(\xh)$.


In order to establish convergence of~(\ref{e:pricing2}), define the Lyapunov function
$$ V(\a):= -\U =  - \sum_i  U_i( \T (\a)) $$
on the compact game domain $\O$ and $\a \in \Real^N$. Taking the derivative of $V$ with respect to time along the pricing dynamics~(\ref{e:pricing2}) yields
$$ 
\begin{array}{ll}
\dot V& =  - \sum_i \sum_j \dfrac{\partial U_i}{\partial x_j} \sum_k  \dfrac{\partial x_j}{\partial \a_k}\, \dot \a_k \\
 & = -\sum_k \left( \sum_i \sum_j  \dfrac{\partial U_i}{\partial x_j}\dfrac{\partial x_j}{\partial \a_k}  \right) \, \dot \a_k \\
 & = -\sum_k (\dot \a_k)^2 \leq 0.
\end{array}
$$
Thus $\dot V =0$ only at $\dot \a_j =0$, $\forall j$, or at its unique equilibrium. Hence, by the LaSalle's theorem, (Theorem 4.4,~\cite{khalilbook}), the pricing scheme~(\ref{e:pricing2}) globally asymptotically converges to its unique equilibrium at which the NE solution coincides with the social maximum.\end{proof}

In the special case of \textit{separable} utility functions $ \U(x):=\sum_i U_i(x_i)$, the pricing mechanism simplifies to
\begin{equation} \label{e:pricing1}
  \dot \a_i = \sum_j \dfrac{\partial U_j}{\partial x_j^*} \dfrac{\partial x_j^*}{\partial \a_i}  \;\; \forall i,
\end{equation}
which ensures that the NE point of the underlying game globally asymptotically converges to the maximum of the social welfare function
under the same condition. 

\subsection*{Example 3}

Consider a game with separable utility functions with the cost 
$$ \begin{array}{l}
 J_i(\a_i,x)= \a_i \left( \sum_i x_i \right) - U_i(x_i),\\ \\
\text{where } U_i:=\beta_i \log(1+x_i) - k_i x_i). 
\end{array} $$
This type of utility function may arise due to inherent physical constraints on player actions such as battery constraints on uplink transmission power levels in wireless devices. Then, the NE solutions is 
$$x_i^* = \dfrac{\beta_i}{\alpha_i+k_i} -1 .$$
Notice that, the matrix $H(\a)$ is diagonal in this case. Furthermore, we can explicitly find 
$$ \dfrac{\partial x^*_i}{\partial \a_i} = -\dfrac{\beta_i}{(\alpha_i+k_i)^2} <0 ,$$
from which non-singularity of $H$ immediately follows. 
The properties of this example also hold for a quadratic pricing function replacing the linear one, i.e., $p_i(x) = \sum_i x^2_i$. However, for $p_i(x) =  e^{\sum_i x_i}$, $x_i^*$ is not independent of $\a_j$ and non-singularity of $H(\a)$ is not immediate.

\subsection*{Example 4}

Consider a variant of Example 1 for the non-separable utility case by defining a  linear pricing and a signal-to-interference ratio (SIR)-like utility function
$$  J_i(\a_i,x)= \a_i \left( \sum_i x_i \right)  - \beta_i \log(1+s_i(x)),$$ 
where \vspace{-0.3cm}
$$s_i(x):=\dfrac{h_i x_i}{\sum_{j \neq i} h_j x_j + \sigma}. $$
Here, the parameter $L=1$ for simplicity.
Then, it follows that under non-singularity  conditions on the system matrix $A$ defined in~(\ref{e:matrix1}), the unique NE is given as $x^* = A^{-1} \; v$ where $v = [ v_i ]$, $v_i = \frac{\beta_i}{\alpha_i} - \sigma$. In this case  $x^*_i$ depends on all pricing parameters, $\a$, and  $H(\a)$ is not diagonal. However, we can still explicitly find $\partial x^* / \partial \a$, and it can be shown that, under non-singularity conditions on $A^T$, $H$ is  non-singular.

\subsection*{Example 5}

An example of dynamic game control is obtained by extending Example 1 to the limited information case where the designer does not have access to user preferences. Then, a dynamic a pricing mechanism can be deployed. Toward this end, define a set of penalty functions $\rho_i(x_i)$ to bring the system within the desired region
\begin{equation} \label{e:barrier1}
  \rho_i (x_i) := 
\begin{cases}
f (\bv_i - (S x)_i ), \text{ if } s_i < \sb_i \\
 0, \text{ else}
\end{cases},
\end{equation}
where the scalar function $f(\cdot)$ is smooth and nondecreasing in its argument, and 
$$f(0)=\dot f(0)=0. $$ 
For example, $f$ could be a quadratic function.

A possible pricing function is then
\begin{equation} \label{e:pricing0}
  \dot \a_i = \sum_j \dfrac{\partial \rho_j}{\partial x_j^*} \dfrac{\partial x_j^*}{\partial \a_i}  \;\; \forall i.
\end{equation}
It is assumed here that the designer (base station) has access to system parameters $L,\, h$, and $\sigma$. The other
terms can be estimated through iterative observations~\cite{kayestimation}.

Finally,  this pricing mechanism ensures that the NE point of the underlying game, $x^*$, enters the desired QoS region $\Oh$.
To show this, define the Lyapunov function 
$$ V:= - \sum_i  \rho_i( x_i ) $$
on the compact game domain $\O$. Taking the derivative of $V$ with respect to time along the pricing dynamics~(\ref{e:pricing0}) yields
$$ \begin{array}{ll}
\dot V& = - \sum_i \dfrac{\partial \rho_i}{\partial x_i} \sum_j  \dfrac{\partial x_i}{\partial \a_j}\, \dot \a_j \\
 & = -\sum_j \left( \sum_i \dfrac{\partial \rho_i}{\partial x_i}   \dfrac{\partial x_i}{\partial \a_j} \right) \, \dot \a_j \\
 & = - \sum_j (\dot \a_j)^2 \leq 0,
\end{array}$$
with $\dot V < 0$ outside the set $\Oh$ and $\dot V=0$ if and only if $\dot \a_i=0 \; \forall i$. Hence, the system converges to the
desired region $\Oh$ under the pricing mechanism.

\subsection{Game and Pricing Dynamics on Two Time-Scales}

The previous analysis has focused on pricing dynamics under the time-scale separation, where game dynamics are assumed to be sufficiently fast. Removing this assumption for a complete treatment, two loops (on two time-scales) will be considered: one outer loop for pricing ($\dot{\a_i}$) and an inner loop for actions ($\dot{x_i}$). The next result presents the full analysis taking both  user and pricing dynamics into account and is based on a \textbf{singular perturbation} or boundary layer approach,~\cite{khalilbook}. Towards this end, Theorem ~\ref{thm:sepconv} is extended to analyze both  pricing (slow) dynamics and user (fast) dynamics by using  a combined Lyapunov function.


\begin{thm} \label{thm:sepconv_twoscales}
Define an objective function  $\U(x):=\sum_i U_i(x)$ which admits a unique inner global maximum $ \xh=\arg \max_{x} \U(x)$ under suitable assumptions for user utilities $U_i\;\forall i$ in a class $\G2$ game. Then, under the pricing mechanism~(\ref{e:pricing2})  the user dynamics~(\ref{e:system1}) globally asymptotically converges to the maximum of the social welfare function, $\xh$, if the two systems are
on separate time-scales, the Jacobian matrix of  the mapping $\T$ with respect to  pricing vector $\a$,
$$H(\a) = \dfrac{\partial x^*}{\partial \a}(\a) = \left [  \dfrac{\partial x^*_i}{\partial \a_j}(\a) \right ], \;\; i,j=1,\dots,N, $$ 
is non-singular and the Jacobian matrix $\Theta$ of  $\dfrac{\partial J_i}{\partial x_i}$,~(\ref{e:system1}), with respect to $x$
$$\Theta(\a,x) = \left [  \dfrac{\partial^2 J_i(\a,x)}{\partial x_j \partial x_i} \right ], \;\; i,j=1,\dots,N, $$ 
 is positive definite. 
\end{thm}

\begin{proof} Consider the pricing dynamics~(\ref{e:pricing2}) or~(\ref{e:pricing1v})   on the slow time-scale $t$, and the user dynamics~(\ref{e:system1}) or~(\ref{e:system2}) on the fast time-scale, $t_f$. Here  $t = \epsilon \, t_f $, with  $\epsilon>0$ a small scaling parameter.  The equilibrium of the full system~(\ref{e:system1}),~(\ref{e:pricing1v}) is described by $\xh = \T (\ah)$ and $\ah = \Th (\xh)$. In \textbf{singular perturbation} form  the full system is written in the shifted coordinates $(\at,\xtb)$ as
\begin{eqnarray} \label{e:sp1}
\dot{\at} &=& H^T(\at+\ah) \left[ \left. \frac{\partial \U (x)}{\partial x}  \right]^T  \right|_{x=\xtb+\xh} \nonumber \\
\epsilon \dot{\xtb} &=&  f(\xtb + \xh) + g(\xtb + \xh) \, (\at + \ah)
\end{eqnarray}
where $\xtb = x - \xh$, $\at = \a - \ah$,  $\dot{\at} = d \at /d t$ and the origin $(\at, \xtb) = (0,0)$ is its equilibrium. 

In the shifted coordinates the \textbf{slow manifold} is defined  as $\xtb = h(\at)$ by taking $\epsilon=0$ and solving on the right-hand side  of~(\ref{e:sp1}), i.e.,
\begin{equation} \nonumber
0 =  f(\xtb + \xh) + g(\xtb + \xh) \, (\a + \ah)
\end{equation}
Since $0 =  f(x) + g(x) \, \a$ for  $ x = \T (\a)$, the foregoing yields $\xtb  + \xh= \T(\at + \ah) $,  so that $\xtb = h(\at)$ for $ h(\at)$ defined by
\begin{equation} \label{e:h_def}
h(\at)  := \T (\at + \ah) - \xh = \T (\at + \ah) - \T(\ah)
\end{equation}
with $h(0) = 0$.

Once the slow manifold is determined, a solution that starts in the manifold is known to stay inside for all times.  In order to analyze the behavior in the neighborhood of this manifold, an additional change of variables is used $y = \xtb - h(\at)$ on the  \textbf{fast dynamics} $\xtb$ with respect to the slow manifold. Using this change of variables the full system  (\ref{e:sp1}) is written in the $(\at, y)$ coordinates as
\begin{eqnarray} \label{e:sp2}
\dot{\at} &=& H^T(\at+\ah) \left[ \left. \frac{\partial \U (x)}{\partial x} \right]^T\right|_{x=y+\T(\at+\ah)}  \\
\epsilon \dot{y} &=&  f(y+\T(\at+\ah)) + g(y+\T(\at+\ah)) \, (\at + \ah) \nonumber\\
& & - \epsilon \frac{\partial h(\at) }{\partial \at} H^T(\at+\ah) \left[ \left. \frac{\partial \U (x)}{\partial x} \right]^T \right|_{x=y+\T(\at+\ah)} \nonumber
\end{eqnarray}
where $y = x - \xh - h(\at)$ and~(\ref{e:h_def}) have been used to obtain $x=y+\T(\at+\ah)$. 

On the slow manifold (for  $y=0$), the dynamics of $\at$ in~(\ref{e:sp2}) define the  reduced system, i.e.,
\begin{equation} \label{e:redsys}
\dot{\at} = H^T(\at+\ah) \left[ \left. \frac{\partial \U (x)}{\partial x} \right]^T \right|_{x=\T(\at+\ah)} 
\end{equation}
On the fast $t_f$ time-scale, the dynamics of $y$ in~(\ref{e:sp2}) define the  boundary-layer system
\begin{equation} \label{e:BLsys}
\frac{d y}{d t_f} =  f(y+\T(\at+\ah)) + g(y+\T(\at+\ah)) \, (\at + \ah) 
\end{equation}
where $ \frac{d y}{d t_f} = \epsilon \dot{ y}$ and $\epsilon =0$ has been used on the right-hand side of~(\ref{e:sp2}).  Using~(\ref{e:system1}), yields component-wise from~(\ref{e:BLsys}) 
\begin{equation} \label{e:BLsys1}
\frac{d y_i}{d t_f} = - \left. \frac{\partial J_i (\at,x)}{\partial x_i} \right|_{x = y+\T(\at+\ah) }
\end{equation} 

To analyze the behavior and stability properties of the full system~(\ref{e:sp2}), an overall \textbf{composite Lyapunov function} is defined based on  Lyapunov functions for the reduced~(\ref{e:redsys}) and boundary-layer systems,~(\ref{e:BLsys1}), respectively.   For the reduced system~(\ref{e:redsys}) the same  Lyapunov function  is used, i.e.,
\begin{equation} \label{e:W1}
V(\at) = - \U (\T(\at + \ah)).
\end{equation}
As shown in the proof of Theorem~\ref{thm:sepconv}, the derivative of $V$ with respect to time along the reduced system dynamics~(\ref{e:redsys}) is negative, 
\begin{equation} \label{e:W1dot}
\dot{V}(\at) \leq -\| \dot{\at} \|^2 \leq 0
\end{equation} 
For the boundary-layer system~(\ref{e:BLsys1}) the following candidate Lyapunov function is defined 
\begin{equation} \label{e:W2} 
W_2 (\at, y) = \frac{1}{2} \sum_i \phi_i^2(\at,y+\T(\at+\ah))
\end{equation}
where
$$
\phi_i(\at,y+\T(\at+\ah)) := \left. \frac{\partial J_i (\at,x)}{\partial x_i} \right|_{x = y+\T(\at+\ah) } 
$$
Note that $\phi_i =  \frac{d y_i}{d t_f}$ and, by using~(\ref{e:BLsys1}),
\begin{eqnarray} \label{e:phidot}
\frac{d \phi_i}{d t_f}  =  \frac{d^2 y_i}{d t^2_f}= - \sum_j \left \{ \frac{\partial}{\partial y_j} \left(   \frac{\partial J_i(\at,x)}{\partial x_i}   \right) \right \} \phi_j
\end{eqnarray}
The derivative of $W_2$ with respect to time along the boundary-layer system trajectory (\ref{e:BLsys}) is given as
$$
\frac{d W_2}{d t_f} (\at,y) = \sum_i \phi_i (\at,y+\T(\at+\ah)) \, \frac{d \phi_i}{d t_f}. 
$$
or, by using~(\ref{e:phidot}),
$$
\frac{d W_2}{d t_f} (\at,y) = - \sum_i  \sum_j \phi_i \Theta_{i,j} \phi_j
$$
where, 
$$
\Theta_{i,j} :=  \frac{\partial^2 J_i (\at,x)}{\partial x_j \partial x_i}
$$
evaluated at $x = y+\T(\at+\ah)$. In vector form this yields
\begin{equation} \label{e:W2dot}
\frac{d W_2}{d t_f} (\at,y)  = - \phi^T \, \Theta \,  \phi <0
\end{equation}
since $\Theta >0$ by assumption. Moreover, for any $\at$ in an appropriate domain  around $0$ an upper bound  that depends only on $y$
can be found for (\ref{e:W2dot}). 

Finally, using the Lyapunov functions for the reduced~(\ref{e:W1}) and boundary-layer systems~(\ref{e:W2}),  the following candidate Lyapunov function is defined 
for the full singularly perturbed system (\ref{e:sp2})
$$W (\at,y)=(1-d)\,V(\at)+d\,W_2 (\at,y)$$
for some  $0<d<1$.  Then based on~(\ref{e:W1dot},~\ref{e:W2dot}) using standard arguments as in the proof of Theorem 11.3 in ~\cite{khalilbook},  it can be shown that there exists an $\epsilon^* >0$ such that for every $0< \epsilon < \epsilon^*$, $W$ is a Lyapunov function for the full  system (\ref{e:sp2})  
and the system is asymptotically stable around the origin.
\end{proof}
\begin{rem}
The bound and required interconnection relations follow immediately if, the stronger condition is imposed that the reduced and boundary-layer systems are exponentially stable as in Theorem 11.4 in~\cite{khalilbook}.
\end{rem}


\subsection*{Example 6}


Consider the Example 2 with cost $J_i$ a special case of~(\ref{e:cost1}), for linear pricing  and OSNR-like utility 
\begin{eqnarray} \label{exIV_C} 
U_i(x) &=& \beta_i (\log \left ( 1+ a_i \frac{\gamma_i(x)} { 1 - \Gamma_{ii} \gamma_i(x) } \right ) - x_i )
\end{eqnarray}
Such an utility that may arise due to specific constraints on individual channel power ensures the existence of a unique maximum for $\U(x) = \sum_i U_i(x)$. 
The Nash equilibrium $x^*$ can be obtained from $\partial J_i/\partial x_i = 0$. 
 If $a_i $ are selected  such that the following diagonal dominance condition holds 
\begin{equation}
a_i > \sum_{j \neq i} \Gamma_{ij} \label{e:NashOSNR}
\end{equation}
then $x^*$ is explicitly expressed as 
$$
x^* = \tilde{\Gamma}^{-1} \, C(\a)
$$
where 
$$\tilde{\Gamma}_{ij} = \left\{ \begin{array}{c} a_i, \hspace{0.6em} \quad i=j \\ \Gamma_{ij}, \quad i \neq j
\end{array} \right. , $$
and $C(\a)$ is a vector with elements $\frac{a_i \beta_i}{\a_i+\beta_i} - n_{0} $. Hence,
\begin{equation} \label{e:exH}
H(\a) = \frac{\partial x^*}{\partial \a} (\a) = \tilde{\Gamma}^{-1} \rm{diag} \left( \frac{-a_i \beta_i}{ (\a_i+\beta_i)^2} \right)
\end{equation}
 is clearly non-singular. The second condition in Theorem \ref{thm:sepconv_twoscales} is that  $\Theta > 0$, where $\Theta$ is for this example:
\begin{eqnarray}
\Theta_{ii} &=& \frac{\partial^2 J_i}{\partial x_i^2} = \frac{a_i^2 \beta_i}{( n_{0} +\sum_{j} \tilde{\Gamma}_{ij} x_j )^2} \nonumber \\
\Theta_{ij} & =& \frac{\partial^2 J_i}{\partial x_j \partial x_i} =
 \frac{a_i \beta_i \Gamma_{ij}}{( n_{0} +\sum_{j} \tilde{\Gamma}_{ij} x_j )^2} \nonumber . 
\end{eqnarray}
It can be immediately seen that if~(\ref{e:NashOSNR}) holds, 
then the matrix $\Theta $ is strictly diagonal dominant. If in addition, $a_i > \sum_{j \neq i} \Gamma_{ji} $ it can be shown that $\Theta^T$ is strictly diagonal dominant, hence $\Theta$ is positive definite. The closed loop system~(\ref{e:pricing2}),~(\ref{e:system1}) is \begin{equation} \label{e:expricingvect}
  \dot \a = H^T(\a) \; \left [\dfrac{\partial \U}{\partial x}(x) \right ]^T
\end{equation}
\begin{equation}
\dot{x} = \bar{f}(x) - \a \label{e:exuservect} \hspace{1.5cm}
\end{equation}
where $H^T(\a)$ is defined in~(\ref{e:exH}), 
and based on~(\ref{exIV_C}),  
the $j^{th}$ element in  $\dfrac{\partial \U}{\partial x}(x)$, for $ \U(x):=\sum_i U_i(x)$, is given as 
\begin{equation}
\dfrac{\partial \U}{\partial x_j}(x) = \bar{f}_j(x) - X_j(x)
\end{equation}
with
\begin{eqnarray}
\bar{f}_j(x) &=& \frac{a_j \beta_j}{n_{0} +  \sum_{k } \tilde{\Gamma}_{jk} x_k } - \beta_j \nonumber \\
X_j(x) &=& \sum_{p \neq j} \frac{ a_p \beta_p \Gamma_{pj} x_p}{ (  n_{0} + \sum_{k \neq p} \Gamma_{pk} x_k )  (n_{0} +(\sum_{k} \tilde{\Gamma}_{pk} x_k )}  \nonumber
\end{eqnarray}
Thus~(\ref{e:expricingvect}),~(\ref{e:exuservect}) represent the closed-loop dynamic system. The objective function $\U$ satisfies $\partial^2 \U/\partial^2 x < 0$ for a nonempty set of the design variables $a$, $\beta$, but the conditions are very complicated and as such  the details are omitted. The following section presents an example based on realistic parameters for the two channel case. 

\section{Simulations} \label{sec:simulation}

The dynamic system described by~(\ref{e:expricingvect}) and~(\ref{e:exuservect}) of Example 6  in the previous section
is simulated numerically in the two players case. Consider an optical fiber link with ten amplifiers, each with a parabolic gain shape according to the formula $G = -4e16 \times (\lambda-1555\times10^{-9})^2  + 15$ dB, where $\lambda$ denotes a channel wavelength,  and a span loss of 10 dB. The system matrix $\Gamma$  is obtained as\begin{equation}
\Gamma = \left[ \begin{array}{cc} 2.47\times 10^{-3} & 2.61\times10^{-3} \\ 2.36\times10^{-3} & 2.5 \times 10^{-3} \end{array} \right] \nonumber
\end{equation}
The following parameters are used: $\beta_i = 1$,  $a=[0.485, \; 0.48]$ such that they satisfy the diagonal dominance condition on $\Gamma$ and $\partial^2 \U/\partial^2 x < 0$. These parameters yield the equilibrium solution $\xh = [0.0134, \; 0.0128]$ milliwatt (mW) and $\ah = [73.4,\; 76.9]$.  An input noise power of $n_{0} =0.43$ nanowatt (nW) is considered.  The closed-loop system~(\ref{e:expricingvect}),~(\ref{e:exuservect}) is simulated for $\epsilon =0.01$, i.e., on two time scales such that  in discrete-time form the pricing algorithm~(\ref{e:expricingvect}) is run every 100 iterations of the user algorithm~(\ref{e:exuservect}).
The user algorithm is implemented in a decentralized way such that only the explicitly measurable OSNR  $\gamma_i$ signal is fed back to the respective channel source. 

The simulations, which take $\a=[18.35 \quad 19.23]$ 
and $x = [0.00043 \hspace{0.5em} 0.00043]^T$ as initial conditions and run for $50$ ($5000$) iterations of the pricing (users) algorithm, show a clear convergence to the equilibrium solution $\xh$, $\ah$ and to the OSNR values of approximately $23$dB  in Figures \ref{fig:x} and\ref{fig:OSNR}). The Figure~\ref{fig:price} depicts the evolution of channel prices, $\a$, as a function of time.
\begin{figure}
  \centering
\includegraphics[width=\columnwidth]{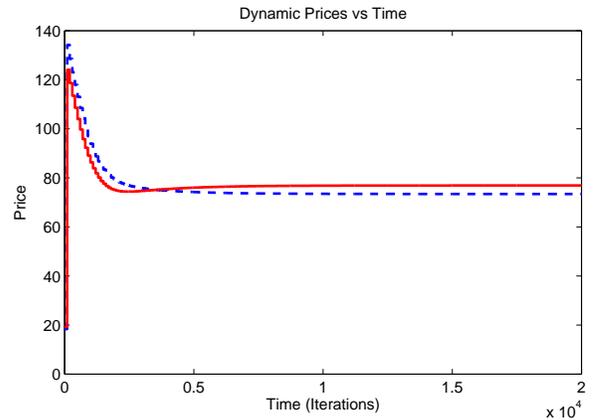}\\
  \caption{Channel prices, $\a$, as a function of time.}\label{fig:price}
\end{figure}

\begin{figure}
  \centering
  \includegraphics[width=\columnwidth]{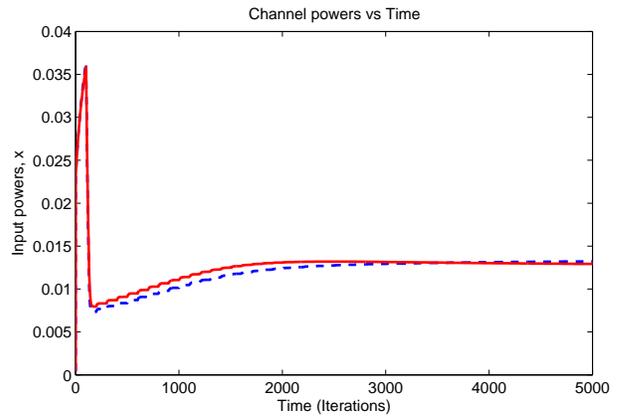}\\
  \caption{Channel powers $x$ ($mW$) as a function of time.}\label{fig:x}
\end{figure}
%
%
\begin{figure}
  \centering
  \includegraphics[width=\columnwidth]{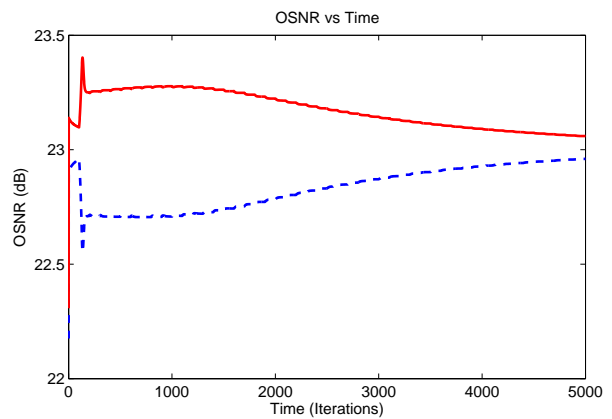}\\
  \caption{OSNR values (dB) as a function of time.}\label{fig:OSNR}
\end{figure}


\section{Conclusion} \label{sec:conclusion}


The  Nash equilibrium design of noncooperative games is discussed from an optimization and control theoretic perspective. 
It is shown that for a fairly general class of noncooperative games,
pricing mechanisms can be used as a design tool to optimize and control the outcome of the game such that certain predefined global objectives are satisfied. The class of games studied are applicable to a number of problems, such as network congestion control, wireless uplink power control, and optical power control. The game design problem is analyzed under full and limited information assumptions, perfect and imperfect time-scale separation between game and pricing dynamics, and both for separable and non-separable utility functions. Sufficient conditions are derived, which allow the designer to place the Nash equilibrium solution or to guide the system trajectory  to a desired region or point. The results obtained are illustrated through a number of examples from wireless and optical
networks.

There are many directions for extending the results presented. One of them is further application of game design methods to specific problems such as power control in optical networks. Others include secure communication and spectrum allocation in wireless networks. Yet another interesting direction is the analysis of estimation error effects on mechanism stability for the limited information case.  In addition, the stability of the combined pricing and game system can be analyzed under time delays, which often have a nonnegligible effect.

\section*{Acknowledgments}

The authors thank Holger Boche for helpful discussions and pointing out the Legendre transform interpretation. This work is in part
supported by Deutsche Telekom Laboratories. Earlier versions of this work have appeared partly in the International Conference on Game Theory
for Networks (GameNets), in Istanbul, Turkey in May 2009~\cite{gamenetsne} and in the Conference on Decision and Control (CDC), Shanghai, China, in December 2009~\cite{cdc09lacra}.

\appendix \label{sec:ne}

This appendix revisits the analysis in \cite{tansuphd,rosen} on existence
and uniqueness of Nash equilibrium under Assumptions~\ref{assm1}-\ref{assm4}.

For games of class $\G1$ and $\G2$, the existence of NE 
follows immediately from a standard theorem of game theory 
(Theorem 4.4, p.176, in~\cite{basargame}) under Assumptions~\ref{assm1} and \ref{assm2}.

In view of Assumption~\ref{assm4}, the Lagrangian function for player $i$ 
in this game is given by
\begin{equation} \label{e:lagrangian}
L_i(x,\mu)=J_i(x)+\sum_{j=1}^r \mu_{i, j} h_j(x) ,
\end{equation} 
where $\mu_{i, j}\geq 0,\;  j=1,\,2,\,\ldots r$ are the Lagrange multipliers of
player $i$~\cite[p. 278]{bertsekas2}.  
We now provide a proposition for $\G2$ games  with conditions similar 
to the well known Karush-Kuhn-Tucker necessary conditions 
(Proposition 3.3.1, p. 310,~\cite{bertsekas2}).

\begin{prop} \label{kuhntucker}
Let $x^*$ be a NE point of a $\G2$ game and Assumptions~\ref{assm1}-\ref{assm4}
hold. There exists then a unique set of Lagrange multipliers, 
$\{\phi_{i, j}:\;  j=1,\,2,\,\ldots r ,\; i=1,\,2,\,\ldots N \} $, such that
$$ \begin{array}{r}
 \dfrac{d L(x^*,\phi)}{d x_i}= \dfrac{d J_i(x^*)}{d x_i}+
 \displaystyle { \sum_{j=1}^r \phi_{i, j}^*\dfrac{d h_j(x^*)}{d x_i}} =0,\\ 
 i=1,\,2,\,\ldots N, \\
\phi_{i, j} \geq 0, \;\; \forall i, j,  \text{ and } \;
\phi_{i, j} = 0, \;\; \forall j \notin A_i(\x^*), \forall i\, ,
\end{array}
$$
where $A_i(x^*)$ is the set of active constraints in $i^{th}$ player's minimization
problem at NE point $x^*$.
\end{prop}

\begin{proof}
The proof essentially follows lines similar to the ones of the
Proposition 3.3.1 of~\cite{bertsekas2}, where the penalty approach
is used to approximate the original constrained problem by
an unconstrained problem that involves a violation of the
constraints. The main difference here is the
repetition of this process for each individual $x_i$ at the NE point $x^*$.
\end{proof}

Define now a more compact notation the vector of Lagrangian functions as 
$L:=[L_1,\ldots,L_N]$, and
the $N \times N$ diagonal matrix of Lagrange multipliers for the $j^{th}$ 
constraint as $\Phi_j=\rm{diag} [\phi_{1,j}, \phi_{2,j}, \ldots \phi_{N,j}]$. 


By Proposition~\ref{kuhntucker} and Assumption~\ref{assm4}, a NE point $x^{(1)}$ satisfies
\begin{equation} \label{e:proofunique1}
 \overline \nabla L(x^{(1)},\Phi^{(1)})=q(x^{(1)})+\sum_{j=1}^r \Phi_j^{(1)} 
\overline \nabla h_j(x^{(1)})=0,
\end{equation}
where $\Phi_j^{(1)}\geq 0$ is unique for each $j$.
Assume there are two different NE points $x^{(0)}$ and $x^{(1)}$. Then, one can also
write the counterpart of (\ref{e:proofunique1}) for $x^{(0)}$. 
Following an argument similar to the one in the proof of Theorem 2 in~\cite{rosen},
one can show that this leads to a contradiction. We present a brief outline of a simplified
version of that proof for the sake of completeness. 

Multiplying (\ref{e:proofunique1}) and its counterpart for $x^{(0)}$ from left by
$(x^{(0)}-x^{(1)})^T$, and then adding them together, we obtain
\begin{equation} \label{e:contradict}
\begin{array}{rcl}
 0 & = & (x^{(0)}-x^{(1)})^T \overline \nabla L(x^{(1)},\Phi^{(1)}) \\
  & & +  \left( \overline \nabla L(x^{(1)},\Phi^{(1)})\right)^T (x^{(0)}-x^{(1)}) \\
  & & + (x^{(1)}-x^{(0)})^T \overline \nabla L(x^{(0)},\Phi^{(0)})  \\ \\
  &= & (x^{(0)}-x^{(1)})^T \left( q(x^{(1)})-  q(x^{(0)}) \right) \\
  & & + \left( q(x^{(1)})-  q(x^{(0)}) \right)^T (x^{(0)}-x^{(1)}) \\
  & & +  (x^{(1)}-x^{(0)})^T \sum_{j=1}^r  [\Phi_j^{(1)} \overline \nabla h_j(x^{(1)}) \\
  & & - \Phi_j^{(0)} \overline \nabla h_j(x^{(0)})] .\\
\end{array}
\end{equation}

Define the strategy vector $x(\theta)$ as a
convex combination of the two equilibrium points $x^{(0)}\,,\,x^{(1)} $ :
$$
  x(\theta)=\theta x^{(1)}  + (1-\theta) x^{(0)}  ,
$$
where $0<\theta<1$. Take the derivative of $q(x(\theta))$ with respect to $\theta$,
\begin{equation}\label{e:gdiff}
  \dfrac{dq(x(\theta))}{d\theta}=Q(x(\theta)) \frac{dx(\theta)}{d\theta}=Q(x(\theta))(x^{(1)} -x^{(0)}),
\end{equation}
where $Q(x)$ is defined in~(\ref{e:g1}). Integrating~(\ref{e:gdiff}) over
$\theta$ yields
\begin{equation}\label{e:gintegral}
  q(x^{(1)})-q(x^{(0)})=\left[\int_0^1 Q(x(\theta)) d\theta \right](x^{(1)}-x^{(0)} ) .
\end{equation}
Multiplying (\ref{e:gintegral})
from left by $(x^{(1)}-x^{(0)} )^T$, the transpose of (\ref{e:gintegral}) from right by
$(x^{(1)}-x^{(0)} )$, and adding these two terms yields
\begin{equation}\label{e:gintegral2}
  (x^{(1)}-x^{(0)})^T \left[\int_0^1 Q(x(\theta))+Q^T(x(\theta)) d\theta \right](x^{(1)}-x^{(0)} ) .
\end{equation}
Since $Q(x(\theta))+Q^T(x(\theta)) $ is positive definite by Assumption~\ref{assm3}
and the sum of two positive definite matrices is positive definite,
the matrix $\bar Q:=\int_0^1 Q(x(\theta))+Q^T(x(\theta))  d\theta$ is positive definite.

Similarly, we have
\begin{equation}\label{e:hdiff}
  \dfrac{d \overline \nabla h(x(\theta))}{d\theta}=H(x(\theta)) \frac{dx(\theta)}{d\theta}=H(x(\theta))(x^1-x^0),
\end{equation}
where $H(x)$ is the Jacobian of $\overline \nabla h(x)$ and positive definite due to convexity of $h(x)$
by definition. The third term in (\ref{e:contradict})
$$ \begin{array}{r}
(x^{(0)}-x^{(1)} )^T \sum_{j=1}^r [\Phi_j^{(0)}\overline  \nabla h_j(x^{(0)})-\Phi_j^{(1)} 
  \overline  \nabla h_j(x^{(1)})],
\end{array}
$$
is less than 
$$ \sum_{j=1}^r [\Phi_j^{(1)}-\Phi_j^{(0)}] [h_j(x^{(1)})-h_j(x^{(0)})] ,$$
due to convexity of $h(x)$. Since for each constraint $j$, $h_j(x) \leq 0\; \forall x$, $\Phi_j^{(i)} h_j(x^{(i)})=0,\; i=0, 1$,  and
$\Phi_j$ is positive definite, where the latter two follow from Karush-Kuhn-Tucker
conditions, this term is also non-positive.

The sum of the first two terms in (\ref{e:contradict}) are the negative of (\ref{e:gintegral2}), which
is strictly positive for all $x^{(1)} \neq x^{(0)}$. Hence, (\ref{e:contradict})
is strictly negative which leads to a contradiction unless $x^{(1)}=x^{(0)}$. Thus, there exists a unique NE point in the class of games $\G2$.

\bibliographystyle{IEEE}

\end{document}